\documentclass[a4paper]{article}

\IfFileExists{xurl.sty}{\usepackage{xurl}}{} 

\usepackage[T1]{fontenc}
\usepackage[utf8]{inputenc}
\usepackage[english]{babel}
\usepackage{latexsym} 
\usepackage{amssymb}
\usepackage{amsmath}
\usepackage{mathpazo} 
\usepackage[dvips]{graphicx} 
\usepackage{amsthm}
\usepackage{xcolor}
\usepackage{multirow}
\usepackage{natbib} 
\usepackage{setspace}
\doublespace

\usepackage[top=2.5cm,bottom=2.5cm,right=3cm,left=2.5cm,headsep=0.4in,includehead]{geometry}

\theoremstyle{plain}
\newtheorem{theorem}{Theorem}[section]

\newtheorem{definition}{Definition}[section]
\newtheorem{example}{Example}[section]

\theoremstyle{definition}

\theoremstyle{remark}

\newcommand{\R}{\mathbb{R}}

\newcommand{\N}{\mathbb{N}}

\newcommand{\SI}{\mathcal{SI}}
\newcommand{\SIW}{\mathcal{SIW}}

\newcommand{\DP}{\mbox{\upshape DP}}
\newcommand{\PG}{\mbox{\upshape PG}}
\newcommand{\ShSh}{\mbox{\upshape SS}}
\newcommand{\CM}{\mbox{\upshape CM}}
\newcommand{\HCM}{\mbox{\upshape HCM}}

\renewcommand{\thefootnote}{\alph{footnote}}

\begin{document}
\renewcommand{\proofname}{{\normalfont\bfseries Proof}}
\title{\textbf{Mergeable weighted majority games and characterizations of some power indices}}
\author{Livino M. Armijos-Toro$^{a,b,e,}$\footnote{Corresponding author: E-mail address: livinomanuel.armijos@rai.usc.es} \and José M. Alonso-Meijide$^{c,e}$ \and Manuel A. Mosquera$^{d,e}$}

\footnotetext[1]{Department of Statistics, Mathematical Analysis, and Optimisation, Faculty of Mathematics, University of Santiago de Compostela, Campus Vida, 15782 Santiago de Compostela, Spain. ORCID: 0000-0001-8553-536X}
\footnotetext[2]{Departamento de Ciencias Exactas, Universidad de las Fuerzas Armadas ESPE, Sangolquí, Ecuador.}
\footnotetext[3]{MODESTYA Research Group, Department of Statistics, Mathematical Analysis, and Optimisation, Faculty of Sciences, University of Santiago de Compostela, Campus de Lugo, 27002 Lugo, Spain. ORCID: 0000-0002-8723-1404}
\footnotetext[4]{Universidade de Vigo, Departamento de Estat\'{i}stica e Investigaci\'{o}n Operativa, 32004 Ourense, Spain. ORCID: 0000-0002-4769-6119}
\footnotetext[5]{CITMAga, 15782 Santiago de Compostela, Spain.}

\date{}
\maketitle

\renewcommand{\thefootnote}{\arabic{footnote}}

\begin{abstract}
	In this paper, we introduce a notion of mergeable weighted majority games with the aim of providing the first characterization of the Colomer-Martínez power index \citep{Colomer1995}. Furthermore, we define and characterize a new power index for the family of weighted majority games that combines ideas of the Public Good \citep{Holler1982} and Colomer-Martínez power indices. Finally, we analyze the National Assembly of Ecuador using these and some other well-known power indices.		
\end{abstract}


\section{Introduction}

The most common way of expressing the will of a group of players regarding the choice of an option or candidate is by voting. Simple games are the appropriate models for studying social and political structures. Power indices could measure the player's influence on decisions made through voting.

It is necessary to introduce tools to measure the power of players as objectively as possible. This measure also has to provide information on the situation of the player within the whole set of players at the time of being part of winning coalitions. A winning coalition can enforce a proposal with which all its members agree. In many cases, a player's power index is based on the number of winning coalitions it belongs to.

The most widely used and studied power indices are the Shapley-Shubik \citep{Shapley1954} and the Banzhaf \citep{Banzhaf1965} power indices. The Shapley-Shubik power index considers all possible permutations (orderings) of all players. Each player is incorporated into the coalition formed by the players preceding it in the permutation. In each permutation, there is a critical player, i. e., a player who changes a losing coalition into a winning one. Considering a uniform distribution over the set of all possible permutations of all players, the Shapley-Shubik power index of a player is the probability that this player is critical. The Banzhaf power index is calculated similarly to the Shapley-Shubik power index, with the difference that the order in which each player joins the coalition is not relevant and, therefore, a uniform distribution over the set of coalitions is considered. The Banzhaf power index does not allocate the total power in the sense that the players' power allocations do not add up to 1. For this reason, it is usual to work with a normalized version of the Banzhaf power index.

For the definition of some power indices, such as the Deegan-Packel \citep{Deegan1978} or Public Good \citep{Holler1982} power indices, only minimal winning coalitions are considered. A minimal winning coalition is a winning coalition that becomes a losing coalition if any player leaves it. The size of the minimal winning coalitions containing a player is relevant to obtain the Deegan-Packel power index. However, the Public Good power index is determined by the number of minimal winning coalitions to which a player belongs, regardless of their size.

An interesting subclass of simple games is the family of weighted majority games. A simple game is a weighted majority game if there are a quota and a weight for each player such that a coalition wins if and only if the sum of the weights of the players in the coalition is greater than or equal to the quota. The National Assembly of Ecuador is an example of a weighted majority game. This assembly has an interesting feature: there is a voting block formed by smaller groups. Thinking of the classical random arrivals model of the Shapley-Shubik index, this would mean that this block arrives together. In that case, it makes sense that the probability assigned to a given group in the block is proportional to its number of members (weights). Therefore, for situations such as the National Assembly of Ecuador, it seems natural to consider power indices that somehow use the weights or, at least, use the contributions or power of each group in a block (or coalition) as proposed in \cite{Koczy2016}.

A power index defined over the class of weighted majority games that uses player weights to distribute total power is the Colomer-Martínez power index \citep{Colomer1995}. Like the Deegan-Packel power index, minimal winning coalitions play a relevant role in the Colomer-Martínez power index. The main difference is that the weight of each player also appears in its calculation. Another power index defined in the class of weighted majority games is the one proposed by \cite{Barua2005}. This power index also uses the weight of each player, but it needs all winning coalitions, unlike the Colomer-Martínez power index. When all players are necessary to make a decision (unanimity), the Colomer-Martínez power index and the power index defined in \cite{Barua2005} allocate power directly proportional to the weight of each player.

The wide range of existing power indices makes the choice of one of them dependent on the context in which it has to be applied. Therefore, studying the properties that each power index satisfies becomes crucial. The mergeability properties will play a fundamental role in this paper. Two simple games are mergeable if no minimal winning coalition of one of the games contains any minimal winning coalition of the other game. Therefore, the set of minimal winning coalitions of the union of mergeable simple games coincides with the union of the sets of minimal winning coalitions of the games. However, this mergeability property of simple games does not work for weighted majority games. A new operation for the union of weighted majority games must be provided to ensure that the resulting game is also a weighted majority game. Consequently, additional conditions have to be established for weighted majority games to be mergeable.
	
It is interesting to study which properties unambiguously characterize a power index, i.e., to provide a set of properties that only a given power index jointly satisfies. This set of properties would allow selecting one power index or another without looking at its formulas. Characterizations of all power indices mentioned above, except for the Colomer-Martínez power index, can be found in the literature. In this paper, the main objective is to provide the first characterization of the Colomer-Martínez power index. The characterization is mathematically compelling and uses well-known properties such as efficiency and null player. It also makes use of a mergeability property similar to that used in characterizations of other power indices. Moreover, the property of weighted symmetry is also necessary. This property requires weighted majority games that have only one minimal winning coalition and compares the power that receives two critical players in relative terms of their weights. A similar property is used to characterize the power index defined in \cite{Barua2005}.

Weighted majority games can refer to private or public goods.  \cite{Holler1983} argued that the Deegan-Packel power index is suitable for working with a private good, therefore the Colomer-Martínez power index, as a weighted version of the Deegan-Packel power index, will also be appropriate for working with a private good. On the other hand, \cite{Holler1983} also pointed out that the Public Good power index is suitable for working with a public good. However, this last power index does not consider the main characteristic of a weighted majority game: its weights. Then, it is interesting to define a power index suitable for working with public goods that considers the players' weights in its calculation. For this reason, the proposed power index uses ideas from both the Public Good and Colomer-Martínez power indices.

The paper is organized as follows. In Section \ref{sec:pre}, some basic definitions of simple games and weighted majority games are given. Section \ref{sec:merg} is devoted to the discussion of a new mergeability property for weighted majority games. In Section \ref{sec:colomer}, a characterization of the Colomer-Martínez power index is provided. In Section \ref{sec:HC}, a new power index for weighted majority games is proposed and characterized. In Section \ref{sec:ejem}, an application of these power indices to the National Assembly of Ecuador is presented. Finally, Section \ref{sec:conclusions} is devoted to conclusions and future work.

\section{Preliminaries}\label{sec:pre}

A \textit{cooperative game with transferable utility}, TU game for short, is a pair $(N,v)$ where $N$ is the set of players and $v: 2^N \longrightarrow \R$ is the characteristic function with $v(\emptyset)=0$. For each $S\subseteq N$, let $|S|=s$ be the cardinality of the set $S$. A TU game is said to be \textit{monotone} when $v(S)\leq v(T)$ for all $S\subseteq T \subseteq N$. The class of \textit{simple games} is the class of monotone TU games such that $v(N) = 1$ and, for each $S \subseteq N$, $v(S) = 0$ or $v(S) =1$. We denote by $\SI$ the class of simple games with an arbitrary set of players.

Let $(N,v)\in \SI$. A coalition $S\subseteq N$ is called a \textit{winning coalition} whenever $v(S)=1$ and is called a \textit{losing coalition} if $v(S)=0$. $W(v)$ denotes the set of winning coalitions of the game $(N,v)$. A coalition $S \subseteq N$ is a \textit{minimal winning coalition} if there does not exist $T\subsetneq S$ such that $T$ is a winning coalition. $M(v)$ denotes the set of minimal winning coalitions of the game $(N,v)$ and let $M_{i}(v) = \{S \in M(v) : i \in S \}$ for each $i\in N$. 

A simple game can represent a voting system. In these systems, it is common to require properness, i.e., the complement of a winning coalition is a losing coalition. Sometimes strongness it is also desirable, i.e., the complement of a losing coalition is a winning coalition. However, these properties are not required here to make this paper as general as possible.

Let $N$ be a set of players and let $S\subseteq N$. The game $(N,u_S) \in \SI$ where the characteristic function is defined, for each $T\subseteq N$, by
\[u_S(T)=\begin{cases}
	1 & \text{if } S\subseteq T\\
	0 & \text{otherwise}
\end{cases}\]
 is called the \textit{unanimity game of coalition $S$}. Note that $|M(u_{S})| = 1$. 

Let $N$ be a set of players and let $(N,v), (N,v') \in  \SI$. The \textit{union game} is the game $(N,v \vee v') \in \SI$ such that, for all $S \subseteq N$, $S \in W(v \vee v')$ if $S \in W(v)$ or $S \in W(v')$;  in other words, $(v \vee v')(S) = \max\{v(S),v'(S)\}$, for each $S\subseteq N$. The \textit{intersection game} is the game $(N,v \wedge v') \in \SI$ such that, for all $S \subseteq N$, $S \in W(v \wedge v') $ if $S \in W(v)$ and $S \in W(v')$; in other words $(v \wedge v')(S) = \min\{v(S),v'(S)\}$ for each $S\subseteq N$. 

Two games $(N,v), (N,v') \in  \SI$ are \textit{mergeable} if for all pair of coalitions $S,T \subseteq N$ such that $S \in M(v)$ and $T \in M(v')$, it holds that $S \nsubseteq T$ and $T \nsubseteq S$. Let us note that if $(N,v)$ and  $(N,v')$ are \textit{mergeable}, then $M(v)\cap M(v')=\emptyset$ and $M(v\vee v')=M(v)\cup M(v')$.

Let $(N,v) \in \SI$. A player $i \in N$ such that $M_{i}(v)= \emptyset$ is called a \textit{null player} in game $(N,v)$. Two players $i,j \in N$ are \textit{symmetric} in the game $(N,v)$ if, for all coalition $S \subseteq N \backslash \{i,j\}$ such that $S \notin W(v)$, $S \cup \{i\} \in W(v)$ if and only if $S \cup \{j\} \in W(v)$. For each player $i \in N$, the set of \textit{swings} of $i$, $\eta_{i}(v)$, is formed by the coalitions $S \subseteq N \backslash \{i\}$ such that $S \notin W(v)$ and $S \cup \{i\} \in W(v)$.

A game $(N,v) \in \SI$ is called a \textit{weighted majority game} if there exist a non-negative vector of weights, $\boldsymbol{w}=(w_{1},w_{2},...,w_{n})$, and a quota $q > 0$ such that $S \in W(v)$ if and only if $w_S \geq q$, where $w_S = \sum_{i \in S} w_{i}$. 

The set $\SIW(N)$ denotes the class of weighted majority games with `$N$ as the set of players. $\SIW$ denotes the class of weighted majority games with an arbitrary set of players. From now on, we identify a weighted majority game by the tuple $[q; w_{1}, w_{2}, ..., w_{n}]$, or $[q;\boldsymbol{w}]$ for short. Note that multiple vectors of weights with their corresponding quotas can result in the same weighted majority game.

The power indices are a solution concept for simple games. A \textit{power index} is a function $f$ which assigns an $n$-dimensional real vector $f(N,v)$ to a game $(N,v) \in \SI$, where, for each player $i\in N$, the $i$-th component of $f(N,v)$ can be interpreted as the power of player $i$ in the game $(N,v)$. 

Let $f$ be a power index on $\SI$. Some desirable properties are \citep[see][]{Alonso-Meijide2013}:

\begin{description}
   \item[EFF:] $f$ satisfies \textit{efficiency} if $\sum_{i \in N} f_{i}(N,v)=1$ for each game $(N,v) \in \SI$.
   \item[NP:] $f$  satisfies the \textit{null player} property  if  $f_{i}(N,v)=0$ for each game $(N,v) \in \SI$ and for each null player $i \in N$ in $(N,v)$.
   \item[SYM:] $f$ satisfies the \textit{symmetry} property if $f_{i}(N,v) = f_{j}(N,v)$  for each game $(N,v) \in \SI$ and for each symmetric players $i,j \in N$ in $(N,v)$.
	\item[TRA:] $f$ satisfies the \textit{transfer} property if $f(N,v \wedge v') +  f(N,v \vee v') = f(N,v) + f(N,v')$ for each pair of games $(N,v), (N,v') \in \SI$.
	\item[DPM:] $f$ satisfies the  \textit{\DP\ - mergeability} property if    
	\begin{equation*}
		f(N,v \vee v' ) = \frac{|M(v)| f(N,v) +|M(v')| f(N,v')}{|M(v \vee v')|} 
	\end{equation*}
	for each pair of mergeable games $(N,v), (N,v') \in \SI$.
	\item[PGM:] $f$ satisfies the \textit{\PG\ - mergeability} property if  
	\begin{equation*}
		f(N,v \vee v' ) = \frac{ f(N,v) \sum_{i \in N} |M_{i}(v) | + f(N,v') \sum_{i \in N} |M_{i}(v') | }{ \sum_{i \in N} |M_{i}(v \vee v')|} 
	\end{equation*}
	for each pair of mergeable games $(N,v), (N,v') \in \SI$. 
\end{description}

Next, some well-known power indices are presented. The \textit{Shapley-Shubik power index} \citep[see][]{Shapley1953,Shapley1954} assigns to each player $i \in N$ the arithmetic mean of the contributions that a player makes to the coalitions previously formed by other players in the $n!$ possible permutations of the players. Formally, for each $(N,v)\in \SI$ and for each $i\in N$,
\begin{equation*}
		\ShSh_{i} (N,v)= \sum_{S \in \eta_{i}(v)} \frac{(s-1)!(n-s)!}{n!}.
\end{equation*}
	
\cite{Dubey1975} characterized the Shapley-Shubik power index by means of EFF, NP, SYM, and TRA.

The \textit{Deegan-Packel power index} \citep{Deegan1978} only considers the coalitions $S \in M_{i}(v)$ for the calculation of the power of a player $i\in N$. This power index assumes that all minimal winning coalitions are equally likely to form and, for each minimal winning coalition, the power is equally divided among its members. Formally, for each $(N,v)\in \SI$ and for each $i\in N$,
\begin{equation*}
		\DP_{i} (N,v)= \frac{1}{|M(v)|} \sum_{S \in M_{i}(v)} \frac{1}{|S|}.
\end{equation*}

\cite{Deegan1978} characterized the Deegan-Packel power index by means of EFF, NP, SYM, and DPM.

The \textit{Public Good power  index} \citep{Holler1982} proposes a share of power proportional to the number of minimal winning coalitions in which each player participates. Formally, for each $(N,v)\in \SI$ and for each $i\in N$,
\begin{equation*}
		\PG_{i}(N,v) =  \frac{|M_{i}(v)|}{\sum_{j \in N} |M_{j}(v)|}.
\end{equation*}
	
\cite{Holler1983} characterized the Public Good power index by means of EFF, NP, SYM, and PGM.

\section{Mergeable weighted majority games}\label{sec:merg}

Note that the three characterizations mentioned in the section above only differ in one property and can be considered parallel. It is interesting to consider parallel characterizations since they reveal similarities and differences among power indices. Such a set of properties that characterize an index allows a researcher to compare a given one with others. Furthermore, they provide information to select the most suitable index depending on the problem to analyze. The transfer and mergeability properties are mathematical appealing and can be interpreted as the effect on a power index of a certain modification of the status of a minimal winning coalition \citep{Laruelle1999}. Thus, the choice of a power index depends on the importance given to minimal winning coalitions.

It is relevant to note that the union of simple games it is not useful for defining mergeability on the class of weighted majority games because the union game of two weighted majority games is sometimes not a weighted majority game. The following example shows it.

\begin{example}
	Let $N = \{1,2,3,4,5\}$, $[q;\boldsymbol{w}]=[13; 8,6,5,1,0]$ and $[q';\boldsymbol{w}']=[15; 1,2,5,10,10]$ be two weighted majority games, and $(N,v)$ and $(N,v')$ their associated simple games, respectively. It is clear that $M(v) = \{\{1,2\}, \{1,3\} \}$ and $M(v') = \{ \{3,4\}, \{3,5\}, \{4,5\} \}$. 
	
	Following the definition of union game for simple games, we get $(N,v\vee v')\in \SI$ such that $M(v \vee v') = \{\{1,2\}, \{1,3\} , \{3,4\}, \{3,5\}, \{4,5\}\}$. Next we show that $(N, v \vee v') $ is not a weighted majority game. 
	
	Assume that there exist weights $w^\vee_{i}$ and a quota $q^\vee$ such that $S \in W(v\vee v')$ if and only if $w^\vee(S) \geq q^\vee$. Then, since $\{1,2\}, \{3,4\} \in M(v \vee v')$,
	\begin{equation*}
			w^\vee_{1} + w^\vee_{2} + w^\vee_{3} + w^\vee_{4} \geq 2q^\vee,
	\end{equation*}
	but, since $\{1,4\}, \{2,3\} \not\in W(v \vee v')$ , then
	\begin{equation*}
		w^\vee_{1} + w^\vee_{4} + w^\vee_{2} + w^\vee_{3} < 2q^\vee,
	\end{equation*}
	which is a contradiction.
\end{example}

There are attempts to define an alternative version of the union game for subclasses of simple games, as in \cite{Barua2005}. Nevertheless, their definition can not be applied to weighted majority games since the resulting union game is sometimes not a weighted majority game.

Therefore, to characterize power indices for weighted majority games using mergeability properties, it is necessary to propose an alternative notion of the union game in the subclass of weighted majority games.

\begin{definition}
	Let $N$ be a set of players, let $p\in \N$ with $p>1$, and let $[q^1;\boldsymbol{w}^1],\ldots, [q^p;\boldsymbol{w}^p] \in \SIW(N)$. The WM-union game of weighted majority games is defined by $[q^\curlyvee;\boldsymbol{w}^\curlyvee] \in \SIW(N)$ where $ q^\curlyvee=\min_{k\in\{1,\ldots, p\}} q^k$ and $\boldsymbol{w}^\curlyvee=\max_{k\in\{1,\ldots, p\}} \boldsymbol{w}^k$.\footnote{If $\boldsymbol{x}^1,\ldots,\boldsymbol{x}^p\in \R^n$, then $\max_{k\in\{1,\ldots, p\}} \boldsymbol{x}^k=(\max_{k\in\{1,\ldots, p\}} x^k_1, \ldots, \max_{k\in\{1,\ldots, p\}} x^k_n)\in \R^n$.}
\end{definition}

Note that, the simple game $(N, v^\curlyvee)$ associated with  $[q^\curlyvee;\boldsymbol{w}^\curlyvee]$ can be different to the game $(N,v\vee v')$, where $(N,v)$ and $(N,v')$ are the simple games associated with $[q;\boldsymbol{w}]$ and $[q';\boldsymbol{w}']$, respectively. For instance, for the set of players $N=\{1,2,3\}$, if we take the weighted majority games $[5;1,2,3]$ and $[6;1,4,5]$, then $[q^\curlyvee;\boldsymbol{w}^\curlyvee]=[5;1,4,5]$. Nevertheless, $\{1,2\}\in M(q^\curlyvee;\boldsymbol{w}^\curlyvee)$ but $\{1,2\}\not\in M(v\vee v')$ since $v(\{1,2\})=0$ and $v'(\{1,2\})=0$ and then $(v\vee v')(\{1,2\})=0$.

The above example also shows that, unlike in the union between simple games, a minimal winning coalition in the WM-union game may not be a minimal winning coalition in any of its component games. For this reason, it is necessary to establish additional conditions for weighted majority games to be mergeable.

\begin{definition}\label{def:mer}
	Let $N$ be a set of players and let $p\in \N$ with $p>1$. It is said that $[q^1;\boldsymbol{w}^1],\ldots, [q^p;\boldsymbol{w}^p] \in \SIW(N)$ are WM-mergeable weighted majority games if 
	\begin{enumerate}
		\item $q^1=q^2=\cdots=q^p$, \label{def:mer_1}
		\item for each player $i \in N$ such that there exist $k,\ell\in\{1,\ldots,p\}$ with  $w^k_i\neq w^\ell_i$, then $w^k_i=0$ or $w^\ell_i=0$, \label{def:mer_2} 
		\item for all $S\subsetneq N$ such that $w^k_S<q^k$ for all $k\in \{1,\ldots,p\}$, then $\sum_{i\in S}\max_{k\in\{1,\ldots, p\}} w_i^k <\min_{k\in\{1,\ldots, p\}} q^k$, \label{def:mer_3}
		\item  $|M(q^\curlyvee;\boldsymbol{w}^\curlyvee)| = \sum_{k=1}^p |M(q^k;\boldsymbol{w}^k)|$. \label{def:mer_4} 		 
	\end{enumerate}
\end{definition}

The first two conditions of the above definition indicate that the quota has to be the same for all weighted majority games, and each player has to have the same weight in each weighted majority game or, failing that, have a weight equal to 0. The following condition states that if a coalition is a losing coalition in all weighted majority games, it is a losing coalition even if each player uses the highest possible weight among its weights in all the weighted majority games. The last condition is the condition for the mergeability of $p$ simple games.

A story that could support the above definition can be the following. Let us think of a parliament that has to take a vote. An absolute majority has to approve the vote, regardless of the number of members of parliament (MPs) present. For various reasons, it may be the case that some parliamentary groups cannot all attend the vote together. Therefore, to take the vote, several sub-votes must be taken with the same quota and in which not all the parliamentary groups are present. The attending parliamentary groups keep their weights (number of MPs), while the absent parliamentary groups have a weight equal to 0. To merge the sub-votes it is reasonable to require that if a coalition of parliamentary groups does not have the number of MPs necessary to approve any of the sub-votes, it also does not have the number of MPs necessary to approve the vote. Moreover, the number of minimal winning coalitions of the vote does not have to be modified.

Note that, the definition of WM-mergeability implies that the set of minimal winning coalitions of the WM-union game is equal to the disjoint union of the sets of minimal winning coalitions of the merged weighted majority games, i. e., $M(q^\curlyvee;\boldsymbol{w}^\curlyvee) = \bigcup_{k=1}^p M(q^k;\boldsymbol{w}^k)$ and $M(q^k;\boldsymbol{w}^k) \cap M(q^\ell;\boldsymbol{w}^\ell)=\emptyset$ for each pair $k,\ell\in\{1,\ldots,p\}$.

\begin{example}\label{ejem:wgame}
	Let $N = \{1,2,3\}$, $[q;\boldsymbol{w}]=[4;3,2,0]$ and $[q';\boldsymbol{w}']=[4;3,0,1]$ be two weighted majority games. It is clear that $M(q;\boldsymbol{w}) = \{\{1,2 \}\}$ and $M(q';\boldsymbol{w}') = \{\{1,3\}\}$, and that the WM-union game is $[q^\curlyvee;\boldsymbol{w}^\curlyvee]=[4;3,2,1]$. Since $M(q^\curlyvee;\boldsymbol{w}^\curlyvee)=\{\{1,2\},\{1,3\}\}$, $w_1=w'_1$, $w'_2=0$, and $w_3=0$, then  $[q;\boldsymbol{w}]$ and $[q';\boldsymbol{w}']$ are WM-mergeable weighted majority games.
\end{example}

Next, based on the DPM and PGM properties for simple games and WM-mergeability, two mergeability properties for power indices in the class of weighted majority games are proposed. Let $f$ be a power index for weighted majority games and let $N$ be a set of players,

\begin{description}
	\item [DPMw:] $f$ satisfies the \textit{\DP\ weighted mergeability} property if, for each $p\in \N$ with $p>1$ and $[q^1;\boldsymbol{w}^1],\ldots, [q^p;\boldsymbol{w}^p] \in \SIW(N)$ WM-mergeable weighted majority games, it holds that
	\begin{equation*}
		f(q^\curlyvee;\boldsymbol{w}^\curlyvee) = \frac{\sum_{k=1}^p |M(q^k;\boldsymbol{w}^k)| f(q^k;\boldsymbol{w}^k)}{|M(q^\curlyvee;\boldsymbol{w}^\curlyvee)|}. 
	\end{equation*}
	\item [HCMw:] $f$ satisfies the \textit{\HCM\ weighted mergeability } property if, for each $p\in \N$ with $p>1$ and $[q^1;\boldsymbol{w}^1],\ldots, [q^p;\boldsymbol{w}^p] \in \SIW(N)$ WM-mergeable weighted majority games, it holds that
		\begin{equation*}
		f(q^\curlyvee;\boldsymbol{w}^\curlyvee) =\frac{\sum_{k=1}^p f(q^k;\boldsymbol{w}^k)\sum_{i \in N}|M_{i}(q^k;\boldsymbol{w}^k)|w^k_{i}}{\sum_{i \in N} |M_{i}(q^\curlyvee;\boldsymbol{w}^\curlyvee)| w_{i}^\curlyvee}.
	\end{equation*}

\end{description}

Note that the DPMw and HCMw properties, like their DPM and PGM counterparts for simple games, say that the power index of the WM-union game is a weighted average of the power indices of its component weighted majority games. Although these properties are purely mathematical conditions, as \cite{Holler1983} pointed out it is not easy to provide a convincing story of their applicability. However, one can think of a parliament as above, where an absolute majority has to approve a vote which must be taken in sub-votes where not all parliamentary groups are present.


\section{Characterization of the Colomer-Martínez power index}\label{sec:colomer}

The main characteristic of weighted majority games is that there exist a quota and weights that define the simple game. Although different quotas and weights can provide the same simple game, it is interesting to use these characteristics to calculate a power index. Nevertheless, the most commonly used power indices for weighted majority games are those based on the associated simple games. The power index defined in \cite{Barua2005} and the Colomer-Martínez power index \citep{Colomer1995} are two exceptions. Both of them use weights for estimating the power of a player. The difference is that the Colomer-Martínez power index uses only minimal winning coalitions meanwhile the power index defined in \cite{Barua2005} uses winning coalitions. This last power index is characterized in \cite{Barua2005}, then this paper focuses on the Colomer-Martínez power index and provides its first characterization.

The Colomer-Martínez power index arises from the idea of building a Government Cabinet, which is why it is also known as the \emph{Execute Power index}. Its formal definition is as follows.

\begin{definition}
	Let  $[q;\boldsymbol{w}] \in \SIW$, the \textit{Colomer-Martínez power index} (\CM) is defined, for each $i\in N$, by
	\begin{equation*}
		\CM_{i}(q;\boldsymbol{w}) =  \frac{1}{|M(q;\boldsymbol{w})|} \sum_{S \in M_{i}(q;\boldsymbol{w})} \frac{w_{i}}{w_S}.
	\end{equation*}

\end{definition}

The Colomer-Martínez power index is an estimation of the expected power of each player in a minimal winning coalition. The power of a player within a coalition is measured as its contribution to the sum of the weights of the players in that coalition. This definition captures the fact that the power of a player will be greater when it joins minimal winning coalitions with players of lesser weight than it. 

Note that the Colomer-Martínez power index is a weighted version of the Deegan-Packel power index in the sense that if we consider a weighted majority game where all weights are equal both power indices coincide. Next example shows the calculation of the Colomer-Martínez power index.

\begin{example}\label{ejem:cmi}
	Let $N = \{1,2,3,4\}$, $[q;\boldsymbol{w}]=[51; 50, 46, 4, 1]\in \SIW(N)$. It is easy to check that $M(q;\boldsymbol{w}) = \{\{1,2\}, \{1,3\}, \{1,4\}, \{2,3,4\}\}$ and then $|M(q;\boldsymbol{w})| = 4$. Moreover, $M_{1}(q;\boldsymbol{w}) = \{\{1,2\},\{1,3\}, \{1,4\}\}$, $M_{2}(q;\boldsymbol{w}) = \{\{1,2\}, \{2,3,4\}\}$, $M_{3}(q;\boldsymbol{w}) = \{\{1,3\}, \{2,3,4\}\}$, and $M_{4}(q;\boldsymbol{w}) = \{\{1,4\}, \{2,3,4\}\}$. 
	
	The Colomer-Martínez power index is 
	\begin{eqnarray*}
		\CM_{1}(q;\boldsymbol{w}) & = & \frac{1}{4} \left(\frac{50}{50 + 46} + \frac{50}{50 + 4} + \frac{50}{50 + 1}\right) \approx 0.6068, \\
		\CM_{2}(q;\boldsymbol{w}) & = & \frac{1}{4} \left(\frac{46}{50 + 46} + \frac{46}{46 + 4 + 1}\right) \approx 0.3453,                 \\
		\CM_{3}(q;\boldsymbol{w}) & = & \frac{1}{4} \left(\frac{4}{50 + 4} + \frac{4}{46 + 4 + 1} \right)\approx 0.0381,                   \\
		\CM_{4}(q;\boldsymbol{w}) & = & \frac{1}{4} \left(\frac{1}{50 + 1} + \frac{1}{46 + 4 + 1} \right)\approx 0.0098.
	\end{eqnarray*}
	Therefore, 
	\begin{equation*}
		\CM(q;\boldsymbol{w}) \approx (0.6068, 0.3453, 0.0381, 0.0098).
	\end{equation*}
\end{example}

It is easy to check that, in Example \ref{ejem:cmi}, the Shapley-Shubik, Deegan-Packel, and Public Good power indices assign the same power to players $2, 3$, and $4$. The reason is that these power indices only consider whether or not the players belong to (minimal) winning coalitions. Nevertheless, it seems natural that if players have different weights, then they should have different power. The Colomer-Martínez power index takes this fact into account and provides a power distribution reflecting the weights of the players.

Next, the first characterization of the Colomer-Martínez power index is provided. For this purpose, some well-known properties adapted to the class of weighted majority games are used, such as EFF and NP. A new weighted symmetry property is proposed. Let $f$ be a power index for weighted majority games,

\begin{description}
	\item [SYMw:] $f$ satisfies the \textit{weighted symmetry} property  if 
		\begin{equation*}
			\frac{f_{i}(q;\boldsymbol{w})}{f_{j}(q;\boldsymbol{w})} = \frac{w_{i}}{w_{j}} 
		\end{equation*}
		for each game $[q;\boldsymbol{w}] \in \SIW$ such that $M(q;\boldsymbol{w}) =\{S\}$ for some $S\subseteq N$, and for all players  $i,j \in S$.
\end{description}

The weighted symmetry property states that, in a weighted majority game where $S$ is the only minimal winning coalition, i.e., in the unanimity game of coalition $S$, the ratio between the power assigned to each pair of players in $S$ has to be the same as the ratio between their weights. A property similar to weighted symmetry is used to characterize the power index defined in \cite{Barua2005}, but SYMw property is weaker since it is only established for unanimity games and players in the minimal winning coalition. 

The following theorem characterizes the Colomer-Martínez power index.

\begin{theorem}\label{teo:cmi}
	The unique power index $f$ on $\SIW$ that satisfies EFF, NP, SYMw, and DPMw is the Colomer-Martínez power index.
\end{theorem}

\begin{proof} 
\textit{Existence}.
First, it is shown that the Colomer-Martínez power index satisfies EFF, NP, SYMw, and DPMw. 

After some algebra, it is clear that the Colomer-Martínez power index satisfies EFF and NP. Moreover, let $[q;\boldsymbol{w}]\in \SIW(N)$ be such that $M(q;\boldsymbol{w}) =\{S\}$ for some $S\subseteq N$ then, for each pair $i,j \in S$ it holds that,
\[
	\frac{\CM_{i}(q;\boldsymbol{w})}{\CM_{j}(q;\boldsymbol{w})}  = \frac{\dfrac{w_{i}}{w_S}}{\dfrac{w_{j}}{w_S}} = \frac{w_{i}}{w_{j}},
\]	                            

and thus the Colomer-Martínez power index also satisfies SYMw.

Next, it is proven that the Colomer-Martínez power index satisfies DPMw. Let  $p\in \N$ with $p>1$ and let $[q^1;\boldsymbol{w}^1],\ldots, [q^p;\boldsymbol{w}^p] \in \SIW(N)$ be WM-mergeable weighted majority games. Then, for each $i\in N$,
\begin{eqnarray*}
	\CM_{i} (q^\curlyvee;\boldsymbol{w}^\curlyvee) & = & \frac{1}{|M(q^\curlyvee;\boldsymbol{w}^\curlyvee)|} \sum_{S \in M_{i}(q^\curlyvee;\boldsymbol{w}^\curlyvee)} \frac{w_{i}^\curlyvee}{w_S^\curlyvee}                                                                \\
	                         & = & \frac{1}{|M(q^\curlyvee;\boldsymbol{w}^\curlyvee)|}  \sum_{k=1}^p  \left[\sum_{S \in M_{i}(q;\boldsymbol{w}^k)} \frac{w^k_{i}}{w^k_S}\right] \\
	                         & = & \frac{\sum_{k=1}^p |M(q^k;\boldsymbol{w}^k)| \CM_{i}(q^k;\boldsymbol{w}^k)}{|M(q^\curlyvee;\boldsymbol{w}^\curlyvee)|},
\end{eqnarray*}
where the second equality follows from the definition of WM-mergeability. 

\textit{Uniqueness}. 
Now, it is shown that the Colomer-Martínez power index is the unique power index that satisfies EFF, NP,  SYMw, and  DPMw properties. 

Let $f$ be a power index on $\SIW$ satisfying all these properties. First, it is proven that $f$ is equal to the Colomer-Martínez power index for weighted majority games with only one minimal winning coalition. Let $N$ be a set of players and let $[q;\boldsymbol{w}]\in \SIW(N)$ be such that $M(q;\boldsymbol{w}) =\{S\}$ for some $S\subseteq N$. It is clear that, for each $i\not\in S$, $i$ is a null player and then, by the NP property it holds that
\begin{equation}\label{thm:uniq1}
	f_i(q;\boldsymbol{w})=0=\CM_i(q;\boldsymbol{w}).
\end{equation}  

Take $i \in S$. By the SYMw property, 
\begin{equation}\label{thm:uniq1_1}
	f_{j}(q;\boldsymbol{w}) = \frac{w_{j}}{w_{i}} f_{i}(q;\boldsymbol{w}) \text{ for all }j \in S.
\end{equation}

Then, by equations (\ref{thm:uniq1}), (\ref{thm:uniq1_1}) and by EFF property:
\begin{equation*}
	1=\sum_{j \in N} f_{j}(q;\boldsymbol{w})=\sum_{j \in S} \frac{w_{j}}{w_{i}} f_{i}(q;\boldsymbol{w}) = \frac{f_{i}(q;\boldsymbol{w})}{w_i}w_S.
\end{equation*}
Therefore, since  $M(q;\boldsymbol{w})=\{S\}$, for each $i\in S$, 
\begin{equation}\label{thm:uniq2}
	f_i(q;\boldsymbol{w}) = \dfrac{w_i}{w_S} =  \CM_i(q;\boldsymbol{w}).
\end{equation}

Now, let $N$ be a set of players and let $[q;\boldsymbol{w}]\in \SIW(N)$ be such that $M(q;\boldsymbol{w}) =\{S_{1}, S_{2},..., S_{m}\}$ with $m > 1$. Let, for all $k\in\{1,\ldots,m\}$, $[q;\boldsymbol{w}^k]\in \SIW(N)$ be such that, for each $i\in N$, 
\begin{equation*}
	w^k_i=\begin{cases}
		w_i & \text{if } i\in S_k \text{ or } M_i(q;\boldsymbol{w})=\emptyset\\
		0 & \text{otherwise.}
	\end{cases}
\end{equation*}
It is easy to check that $M(q;\boldsymbol{w}^k) = \{S_k\}$. Moreover, for all $k\in \{1,\ldots,m\}$, the weighted majority games $[q;\boldsymbol{w}^k]$ are WM-mergeable and their WM-union game $[q^\curlyvee;\boldsymbol{w}^\curlyvee]$ is such that
\begin{eqnarray*}
	q^\curlyvee & = & \min\{q,q,...,q\}= q  \\
	\boldsymbol{w}^\curlyvee & = & \max_{k\in\{1,\ldots, m\}}\{\boldsymbol{w}^k\}= (\max\{w^1_1, \ldots ,w^m_1\},\ldots,\max\{w^1_n, \ldots ,w^m_n\})= (w_{1}, w_{2},..., w_{n})
\end{eqnarray*}
i. e., $[q^\curlyvee;\boldsymbol{w}^\curlyvee]=[q;\boldsymbol{w}]$.

Therefore, for each $i\in N$,
\begin{eqnarray*}
	f_{i}(q;\boldsymbol{w}) & = & f_{i}(q^\curlyvee;\boldsymbol{w}^\curlyvee)                                                      \\
	           & = & \frac{\sum_{k = 1}^{m} |M(q;\boldsymbol{w}^k)| f_{i}(q;\boldsymbol{w}^k)}{|M(q^\curlyvee;\boldsymbol{w}^\curlyvee)|} \\
	           & = & \frac{\sum_{k = 1}^{m} |M(q;\boldsymbol{w}^k)| \CM_{i}(q;\boldsymbol{w}^k)}{ |M(q^\curlyvee;\boldsymbol{w}^\curlyvee)|}       \\
	           & = & \CM_{i}(q^\curlyvee;\boldsymbol{w}^\curlyvee)\\
	           & = & \CM_{i}(q;\boldsymbol{w}),
\end{eqnarray*}
where the second equality follows from applying the DPMw property of $f$, the third one is a consequence of (\ref{thm:uniq1}) and (\ref{thm:uniq2}) since, for each $k\in \{1,\ldots,m\}$,  $[q;\boldsymbol{w}^k]$ are weighted majority games with only one minimal winning coalition, and the fourth equality follows from applying the DPMw property of the Colomer-Martínez power index.

Therefore, the Colomer-Martínez power index is the only power index that satisfies the properties EFF, NP, SYMw, and DPMw.
\end{proof} 

Next, the logical independence of the properties used in Theorem \ref{teo:cmi} is shown.

\begin{description}
	\item [EFF:] Let $f$ be the power index defined, for all $[q; \boldsymbol{w}] \in \SIW$ and $i \in N$,	by
	\begin{equation*}
		f_{i}(q; \boldsymbol{w}) =  2\cdot \CM_i(q; \boldsymbol{w}).
	\end{equation*}
	
	It is clear that $f$ satisfies NP, SYMw, and DPMw, since the Colomer-Martínez power index satisfies them, but $f$ does not satisfy EFF since, for all	$[q; \boldsymbol{w}] \in \SIW$,  it holds that $\sum_{i \in N} f_{i}(q; \boldsymbol{w})=  \sum_{i \in N} 2\cdot \CM_i(q; \boldsymbol{w})	= 2$, 
	since the Colomer-Martínez power index satisfies EFF.
	
	\item [NP:] Let $f$ be the power index defined, for all $[q; \boldsymbol{w}] \in \SIW$ and $i \in N$, by
	
	\begin{equation*}
		f_{i}(q; \boldsymbol{w}) =  \begin{cases}
			\frac{1}{3} & \text{if } [q; \boldsymbol{w}]=[4;2,2,1],\\
			\CM_i(q; \boldsymbol{w})& \text{otherwise.}
		\end{cases}
	\end{equation*}

	It is clear that $f$ satisfies the EFF property, since the Colomer-Martínez power index satisfies it. $f$ also satisfies the SYMw property. Namely, if $[q; \boldsymbol{w}]=[4;2,2,1]$, then $M(q; \boldsymbol{w}) = \{\{1,2\}\}$ and, for all $i,j \in \{1,2\}$ it holds that		
	\[
	\frac{f_{i}(4;2,2,1)}{f_{j}(4;2,2,1)} =  \frac{\frac{1}{3}}{\frac{1}{3}} = 1 = \frac{w_{i}}{w_{j}}. 
	\]	
	On the other hand, for all $[q; \boldsymbol{w}]\neq [4;2,2,1]$, since $f(q; \boldsymbol{w})=\CM(q; \boldsymbol{w})$ and the Colomer-Martínez power index satisfies SYMw, it holds that
	\[
	\frac{f_{i}(q; \boldsymbol{w})}{f_{j}(q; \boldsymbol{w})} =  \frac{w_{i}}{w_{j}}. 
	\]
	
	Moreover, $f$ satisfies the DPMw property. It is clear that $[4;2,2,1]$ cannot be the WM-union game of weighted majority games since it has only one minimal winning coalition. Moreover, $[4;2,2,1]$ may also not be WM-mergeable with other weighted majority games since any WM-union of weighted majority games where $[4;2,2,1]$ is one of the games to be merged turns out to be the weighted majority game $[4;2,2,1]$ itself. Then, $f$ satisfies DPMw since the Colomer-Martínez power index satisfies it.
	
	Nevertheless, $f$ does not satisfy NP since $3$ is a null player of the weighted majority game $[4;2,2,1]$, but 
\[
f_3(4;2,2,1) = \frac{1}{3} \neq 0.
\]

	\item [SYMw:] The Deegan-Packel power index satisfies EFF, NP, and DPM, as shown in \cite{Deegan1978}. Moreover, if several weighted majority games are WM-mergeable, their associated simple games are mergeable (with the definition for simple games). Hence, by DPM, the Deegan-Packel power index satisfies DPMw. However, it does not hold SYMw, or else Theorem \ref{teo:cmi} fails.
	
	\item [DPMw:] It will be shown in the next section that \HCM\ power index satisfies EFF, NP, and SYMw. Nevertheless, it does not satisfy DPMw as it will be shown in Example \ref{ejem:HCM_DPMw}. 	
\end{description}

\section{The Holler-Colomer-Martínez power index}\label{sec:HC}

In this section a new power index for weighted majority games where the weights of the players are directly involved is proposed. \cite{Holler1983} argued that the Deegan-Packel power index is suitable for working with a private good. Therefore, the Colomer-Martínez power index, as a weighted version of the Deegan-Packel power index, will also be appropriate for working with a private good. There are situations in which weighted majority games refer to public goods and \cite{Holler1983} proposed to use the Public Good power index. However, this power index does not consider the weights in weighted majority games. For this reason, it is interesting to define a power index that is suitable for working with public goods and that considers the weights of the players in its calculation. This new power index combines the Public Good and Colomer-Martínez power indices and will be called the Holler-Colomer-Martínez power index (\HCM\ power index, for short). 

For the calculation of the \HCM\ power index, the contribution of each player to the total number of minimal winning coalitions to which it belongs is considered, like the Public Good power index, and also the weights of each player, like the Colomer-Martínez power index. Next, the formal definition of the \HCM\ power index and a characterization using a mergeability property are provided.

\begin{definition}
	Let  $[q;\boldsymbol{w}] \in \SIW$, the \textit{\HCM\ power index} is defined, for each $i\in N$, by
	\begin{equation*}
		\text{\HCM}_{i}(q;\boldsymbol{w}) =   \frac{|M_{i}(q;\boldsymbol{w})|w_i}{\sum_{j \in N} |M_{j}(q;\boldsymbol{w})|w_j}. 
	\end{equation*}
\end{definition}

The \HCM\ power index calculates the power of each player as the proportion, weighted by the weights in the game, of the number of minimal winning coalitions to which it belongs. Alternative expressions for the \HCM\ power index are, for each $i\in N$,
\begin{equation*}
	\text{\HCM}_{i}(q;\boldsymbol{w}) = \frac{\sum_{T \in M_{i}(q;\boldsymbol{w})} w_{i}}{\sum_{j \in N} \sum_{S \in M_{j}(q;\boldsymbol{w})} w_{j}} = \sum_{T \in M_{i}(q;\boldsymbol{w})}\frac{ w_{i}}{\sum_{S \in M(q;\boldsymbol{w})} w_S}.
\end{equation*}

Note that the \HCM\ power index is also a weighted version of the Public Good power index in the sense that if we consider a weighted majority game with all weights equal both power indices coincide. The following example shows the calculation of the \HCM\ power index for a weighted majority game.

\begin{example}\label{ejem:HC}
	Considering the weighted majority game proposed in Example \ref{ejem:cmi}. Recall that $M_{1}(q;\boldsymbol{w}) = \{\{1,2\},\{1,3\}, \{1,4\}\}$, $M_{2}(q;\boldsymbol{w}) = \{\{1,2\}, \{2,3,4\}\}$, $M_{3}(q;\boldsymbol{w}) = \{\{1,3\}, \{2,3,4\}\}$, and $M_{4}(q;\boldsymbol{w}) = \{\{1,4\}, \{2,3,4\}\}$, then 
	\begin{eqnarray*}
		\text{\HCM}_{1}(q;\boldsymbol{w}) &=& \frac{3\cdot 50}{3\cdot 50 + 2\cdot 46 + 2\cdot 4 + 2\cdot 1} = \frac{150}{252} \approx 0.5953, \\
		\text{\HCM}_{2}(q;\boldsymbol{w})  &=&  \frac{2\cdot 46}{3\cdot 50 + 2\cdot 46 + 2\cdot 4 + 2\cdot 1} = \frac{92}{252} \approx 0.3651, \\
		\text{\HCM}_{3}(q;\boldsymbol{w})  &=&  \frac{2\cdot 4}{3\cdot 50 + 2\cdot 46 + 2\cdot 4 + 2\cdot 1} = \frac{8}{252} \approx 0.0317, \\
		\text{\HCM}_{4}(q;\boldsymbol{w})  &=&  \frac{2\cdot 1}{3\cdot 50 + 2\cdot 46 + 2\cdot 4 + 2\cdot 1} = \frac{2}{252} \approx 0.0079.
	\end{eqnarray*}
	Therefore, the \HCM\ power index is 
	\begin{equation*}
		\text{\HCM}(q;\boldsymbol{w}) \approx (0.5953, 0.3651, 0.0317, 0.0079).
	\end{equation*}
\end{example}

The \HCM\ power index calculates the power of each player as the sum, on the minimal winning coalitions to which it belongs, of its contribution to the total weight of all minimal winning coalitions. The Colomer-Martínez power index, however, first calculates the power of each player within each of the minimal winning coalitions to which it belongs and then averages these powers.

Next example shows that the \HCM\ power index does not satisfy DPMw property.

\begin{example}\label{ejem:HCM_DPMw}
	Let $[q; \boldsymbol{w}]=[4;3,2,0]$ and $[q'; \boldsymbol{w}']=[4;3,0,1]$ which are WM-mergeable (see Example \ref{ejem:wgame}). Then, $[q^{\curlyvee};\boldsymbol{w}^{\curlyvee}]=[4;3,2,1]$ and 
	\[\HCM(4;3,2,1)=\left(\frac{2\cdot3}{2\cdot3+2+1},\frac{2}{2\cdot3+2+1} ,\frac{1}{2\cdot3+2+1} \right)=\left(\frac{6}{9},\frac{2}{9},\frac{1}{9}\right).\]
	Nevertheless,
	\begin{eqnarray*}
		\frac{|M(q;\boldsymbol{w})| \HCM(q; \boldsymbol{w}) + |M(q';\boldsymbol{w}')| \HCM(q'; \boldsymbol{w}')}{|M(q^{\curlyvee};\boldsymbol{w}^{\curlyvee})|}&=& \frac{1\cdot \left(\frac{3}{3+2+0},\frac{2}{3+2+0},0\right)+1\cdot\left(\frac{3}{3+0+1},0,\frac{1}{3+0+1}\right)}{2}\\
		&=&\left(\frac{27}{40},\frac{8}{40},\frac{5}{40}\right).
	\end{eqnarray*}
\end{example}

Note that if $p$ weighted majority games are WM-mergeable, with $p>1$, then, for each $i\in N$, $M_{i}(q^\curlyvee;\boldsymbol{w}^\curlyvee)=\bigcup_{k=1}^p M_i(q^k;\boldsymbol{w}^k)$ and $M_i(q^k;\boldsymbol{w}^k) \cap M_i(q^\ell;\boldsymbol{w}^\ell)=\emptyset$ for each pair $k,\ell\in\{1,\ldots,p\}$, and furthermore for all $k\in\{1,\ldots,p\}$ such that $M_i(q^k;\boldsymbol{w}^k)\neq \emptyset$ then $w^k_i=w^\curlyvee_i$. Therefore,  $|M_{i}(q^\curlyvee;\boldsymbol{w}^\curlyvee)|w^\curlyvee_i=\sum_{k=1}^p  |M_{i}(q^k;\boldsymbol{w}^k)|w^k_i$ for each $i\in N$ and thus the \HCM\ power index satisfies HCMw.

Using the HCMw property and others, the following characterization of the \HCM\ power index for weighted majority games is proposed. 

\begin{theorem}\label{theo:HC}
	The unique power index $f$ on $\SIW$ that satisfies EFF, NP, SYMw, and HCMw is the \HCM\ power index. 
\end{theorem}

\begin{proof}
	It immediately follows from a similar reasoning to that in Theorem \ref{teo:cmi}.
\end{proof}

Next, the logical independence of the properties used in Theorem \ref{theo:HC} is shown.

\begin{description}
	\item [EFF:] Let $f$ be the power index defined, for all $[q; \boldsymbol{w}] \in \SIW$ and $i \in N$, by
	\begin{equation*}
		f_{i}(q; \boldsymbol{w}) =  2\cdot \HCM_i(q; \boldsymbol{w}).
	\end{equation*}
	
	It is clear that $f$ satisfies NP, SYMw, and HCMw, since $\HCM$ satisfies them, but $f$ does not satisfy EFF since, for all $[q; \boldsymbol{w}] \in \SIW$,  it holds that $\sum_{i \in N} f_{i}(q; \boldsymbol{w})=  \sum_{i \in N} 2\cdot \HCM_i(q; \boldsymbol{w})	= 2$ since $\HCM$ satisfies EFF.
	
	\item [NP:] Let $f$ be the power index defined, for all $[q; \boldsymbol{w}] \in \SIW$ and $i \in N$, by	
	\begin{equation*}
		f_{i}(q; \boldsymbol{w}) =  \begin{cases}
			\frac{1}{3} & \text{if } [q; \boldsymbol{w}]=[4;2,2,1],\\
			\HCM_i(q; \boldsymbol{w})& \text{otherwise.}
		\end{cases}
	\end{equation*}

	It is clear that $f$ satisfies the EFF property, since \HCM\ satisfies it. $f$ also satisfies the SYMw property. Namely, if $[q; \boldsymbol{w}]=[4;2,2,1]$, then $M(q; \boldsymbol{w}) = \{\{1,2\}\}$ and, for all $i,j \in \{1,2\}$ it holds that		
	\[
	\frac{f_{i}(4;2,2,1)}{f_{j}(4;2,2,1)} =  \frac{\frac{1}{3}}{\frac{1}{3}} = 1 = \frac{w_{i}}{w_{j}}. 
	\]
	
	On the other hand, for all $[q; \boldsymbol{w}]\neq [4;2,2,1]$, since $f(q; \boldsymbol{w})=\HCM(q; \boldsymbol{w})$ and the \HCM\ power index satisfies SYMw, it holds that
	\[
	\frac{f_{i}(q; \boldsymbol{w})}{f_{j}(q; \boldsymbol{w})} =  \frac{w_{i}}{w_{j}}. 
	\]
	
	Moreover, $f$ satisfies the HCMw property. 	We know from the previous section that $[4;2,2,1]$ cannot be the WM-union game of weighted majority games it may also not be WM-mergeable with other weighted majority games. Then $f$ satisfies HCMw since the \HCM\ power index satisfies it.
	
	Nevertheless, $f$ does not satisfy NP since $3$ is a null player of the weighted majority game $[4;2,2,1]$, but 
	\[
	f_{3}(4;2,2,1) = \frac{1}{3} \neq 0.
	\]

	\item [SYMw:] Let $f$ be the power index defined, for all $[q; \boldsymbol{w}] \in \SIW$ and $i \in N$, by	
	\begin{equation*}
		f_{i}(q; \boldsymbol{w}) =  \begin{cases}
			1 & \text{if } [q; \boldsymbol{w}]=[4;2,2,1] \text{ and } i=1,\\
			0 & \text{if } [q; \boldsymbol{w}]=[4;2,2,1] \text{ and } i\neq 1,\\
			\HCM_i(q; \boldsymbol{w})& \text{otherwise.}
		\end{cases}
	\end{equation*}
	$f$ satisfies the EFF property since \HCM\ satisfies it. $f$ also satisfies NP since \HCM\ satisfies it and $f_{3}(4;2,2,1) =  0$ being $3$ the unique null player of the weighted majority game $[4;2,2,1]$. Moreover, $f$ satisfies the HCMw property and is shown analogously to that performed with the power index used in the logical independence of the NP property.
	
	Nevertheless, $f$ does not satisfy the SYMw property. Namely, we know that $M(4;2,2,1) = \{\{1,2\}\}$ and it holds that		
	\[
	\frac{f_{2}(4;2,2,1)}{f_{1}(4;2,2,1)} =  \frac{0}{1} = 0 \neq \frac{w_{i}}{w_{j}}. 
	\]

	\item [HCMw:] It has been shown in the previous section that the Colomer-Martínez power index satisfies EFF, NP, and SYMw properties. Nevertheless, it does not fulfill HCMw, or else Theorem \ref{theo:HC} fails. 
\end{description}

\section{The National Assembly of Ecuador}\label{sec:ejem}

The National Assembly of Ecuador (Asamblea Nacional de Ecuador) has 137 assembly members (MPs from now on). Ecuador held general elections in February 2021.\footnote{\url{https://www.primicias.ec/noticias/politica/los-cambios-en-las-bancadas-de-la-asamblea/}, last accessed 23/12/2021.} The Assembly was composed of:\footnote{The numbers in parentheses represent the number of MPs of each political party.} (49) UNES, (27) MUPP, (18) ID, (18) PSC, (12) CREO, and the minorities (IND): (2) AVA, (2) MEU, (2) AH, (1) PSP, (1) AU, (1) MAP, (1) MUE, (1) MMI, (1) MAE, and (1) DEMSI. The CREO and PSC parties are in favor of the free market. The ID and UNES parties are progressive parties. Finally, the MUPP party is closer to socialism. The current president of Ecuador belongs to the CREO party. 

The so-called legislative benches are the official political groupings within the Assembly, with the right to have authority within the different legislative commissions. Table \ref{Tb01} and Figure \ref{Fig00} summarize the number of MPs on each legislative bench.

\begin{table}[!h]
	\begin{center}
		\begin{tabular}{ r|c }
			\hline
			\textbf{Benches} & \textbf{Assembly members} \\ \hline
			UNES    &        49        \\
			MUPP    &        27        \\
			ID      &        18        \\
			PSC     &        18        \\
			CREO    &        12        \\
			IND     &        13        \\ \hline
		\end{tabular}
	
	\caption{Composition of the National Assembly of Ecuador in May 2021.}	\label{Tb01}
	\end{center}
	
\end{table}

\begin{figure}[ht]
	\centering
	\includegraphics[width=0.7\linewidth]{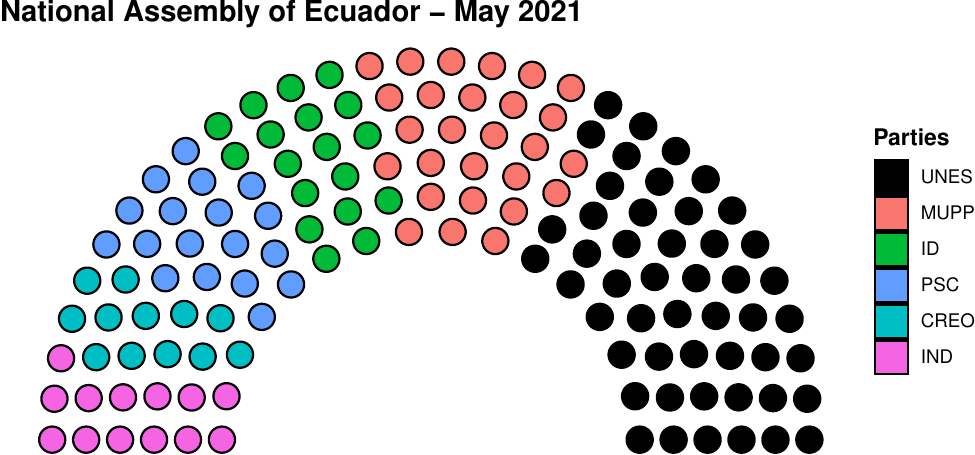} 
	
	\caption{Legislative benches of the National Assembly of Ecuador in May 2021.}\label{Fig00}
\end{figure}

During a legislature and due to splits within the legislative benches, some parliamentary groups may disappear, but new parliamentary groups may also be formed. This fact raises the interest in analyzing the evolution of power within this Assembly throughout a legislature. In June 2021, the party closely aligned to the government (CREO), and some MPs from other parties, mainly minorities, consolidated the new legislative bench (25) BAN. Moreover, some other MPs declared themselves independent: (9) IND. Figure \ref{Fig01} shows the redistribution of MPs in the legislative benches as of June 2021.

\begin{figure}[ht]
	\centering
	\includegraphics[scale=0.8]{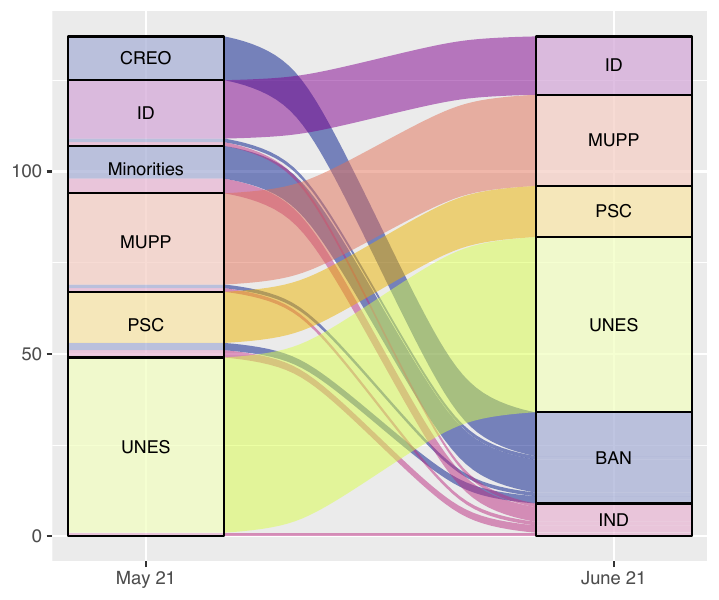} 
	
	\caption{Changes in the legislative benches of the National Assembly of Ecuador. May-June 2021.} \label{Fig01}
\end{figure}

There have been changes in the structure of the legislative benches of the Assembly during the first seven months of its operation. These changes are shown in Table \ref{Tb02} and correspond to June 2021, July 2021\footnote{\url{https://www.primicias.ec/noticias/politica/bancadas-pierden-miembros-votos-asamblea/}, last accessed 23/12/2021.}, 12 October 2021\footnote{\url{https://www.primicias.ec/noticias/politica/union-unes-pachakutik-debilidad-legislativa-gobierno/}, last accessed 23/12/2021
}, 26 October 2021\footnote{\url{https://www.primicias.ec/noticias/politica/posible-destitucion-lasso-apoyo-asamblea/}, last accessed 23/12/2021}, and December 2021\footnote{\url{https://www.primicias.ec/noticias/politica/ruptura-pachakutik-capitulo-bancadas-desgranads/}, last accessed 23/12/2021
}.

\begin{table}
	\begin{center}
		\begin{tabular}{ r|ccccc}
				\hline
				\textbf{Benches} & \textbf{Jun 21} & \textbf{Jul 21} & \textbf{12 Oct 21} & \textbf{26 Oct 21} & \textbf{Dec 21} \\ \hline
				UNES    & 48     &   47   & 47        & 47        & 47     \\
				MUPP    & 25     &   24   & 25        & 25        & 25     \\
				BAN     & 25     &   25   & 25        & 26        & 28     \\
				ID      & 16     &   16   & 14        & 14        & 14     \\
				PSC     & 14     &   14   & 14        & 14        & 14     \\
				IND     & 9      &   11   & 12        & 11        & 9      \\ \hline
		\end{tabular}
	
	\caption{Changes in the legislative benches of the National Assembly of Ecuador. June-December 2021.}
	\label{Tb02}	
	\end{center}
\end{table}

The Assembly can be analyzed using weighted majority games where the quota q is\footnote{One would expect a quota equal to $69$. However, according to the electoral law of Ecuador (art. 8 of the Organic Law of the Legislative Function as amended on November 10, 2020), the simple majority shall be understood as a favorable vote of half plus one of the MPs. Moreover, this law also indicates that if the result of the latter calculation is not an integer number, then the number of votes required is the nearest higher integer number. Therefore, out of $137$ MPs, half plus one would be $69.5$, and hence the quota would be $70$.} $70$ and the vectors of weights  $\boldsymbol{w}$ are given by the number of MPs that each legislative bench has in the different compositions of the Assembly. For the Assembly in May 2021, it can be represented as the following weighted majority game: 
\begin{equation*}
	[q;\boldsymbol{w}] = [70; 49, 27, 18, 18, 12, 13],
\end{equation*}
with $N = \{$UNES, MUPP, ID, PSC, CREO, IND$\}$. Table \ref{Tb04.1} shows the 11 minimal winning coalitions in this weighted majority game.

\begin{table}[ht]
	\begin{center}
		\begin{tabular}{ l }
				\hline
				\textbf{Minimal winning coalitions} \\ \hline
				\{UNES, MUPP\}             \\
				\{UNES, ID, PSC\}          \\
				\{UNES, ID, CREO\}         \\
				\{UNES, ID, IND\}          \\
				\{UNES, PSC, CREO\}        \\
				\{UNES, PSC, IND\}         \\
				\{UNES, CREO, IND\}        \\
				\{MUPP, ID, PSC, CREO\}    \\
				\{MUPP, ID, PSC, IND\}     \\
				\{MUPP, ID, CREO, IND\}    \\
				\{MUPP, PSC, CREO, IND\}   \\ \hline
		\end{tabular}
	
	\caption{Minimal winning coalitions for the weighted majority game associated with National Assembly of Ecuador in May 2021.} \label{Tb04.1}	
	\end{center}
\end{table}

Looking at the minimal winning coalitions in Table \ref{Tb04.1}, one can see that this composition of the Assembly has no null players. Note that the same holds for all Assembly configurations in 2021. One can also see that, in May 2021, UNES is the legislative bench with the most power in the Assembly, in the sense that it belongs to 7 out of 11 minimal winning coalitions. The legislative bench of the president of Ecuador, CREO, also belongs to 3 of these 7 minimal winning coalitions. Nevertheless, these coalitions will be practically unfeasible due to the incompatibility of political ideas between these political parties. Therefore in practice, CREO can only participate in minimal winning coalitions of cardinality equal to 4. This fact may be one of the reasons why the Government has to make a great effort to achieve the approval of its laws in the Assembly.

Table \ref{Tb03} shows the power that each of the power indices discussed in this article assigns to each legislative bench in the May 2021 Assembly composition. All of them give the most power in the Assembly to UNES. Moreover, CREO is one of the legislative benches that obtains the least power with any of the power indices used. Note also that the Shapley-Shubik, Deegan-Packel, and Public Good power indices are the ones that show a more homogeneous distribution of power among the legislative benches of ID, PSC, CREO, and IND since these are symmetric players. On the other hand, the differences with respect to the Colomer-Martínez and \HCM\ power indices are due to the use in their calculation of the number of MPs (weights) that each legislative bench has. It can also be observed that the Colomer-Martínez and \HCM\ power indices distribute total power in a similar way to the Shapley-Shubik power index.

\begin{table}[ht]
	\centering
	
	\begin{tabular}{ r| c  c  c  c  c  c }
		\hline
		\textbf{Power indices} & \textbf{UNES} & \textbf{MUPP} & \textbf{ID} & \textbf{PSC} & \textbf{CREO} & \textbf{IND} \\ \hline
		\ShSh                 & 0.4000        & 0.2000        & 0.1000      & 0.1000       & 0.1000        & 0.1000       \\
		\DP                   & 0.2273        & 0.1364        & 0.1591      & 0.1591       & 0.1591        & 0.1591       \\
		\PG                   & 0.1944        & 0.1389        & 0.1667      & 0.1667       & 0.1667        & 0.1667       \\
		\CM                   & 0.3954        & 0.1675        & 0.1271      & 0.1271       & 0.0881        & 0.0948       \\
		\HCM                  & 0.4064        & 0.1600        & 0.1280      & 0.1280       & 0.0853        & 0.0924       \\ \hline
	\end{tabular}

\caption{Power indices for the weighted majority game associated with National Assembly of Ecuador in May 2021.} \label{Tb03}
\end{table}

Finally, the changes that occur in the distribution of power of the legislative benches due to the movements in the composition of the Assembly (see Table \ref{Tb02}) are analyzed. Consider each of the legislative bench structures in Table \ref{Tb02} as a weighted majority game. All these games have the same set of minimal winning coalitions (Table \ref{Tb04}). One can observe that UNES remains the legislative bench with the most power in the Assembly in the sense that it belongs to 5 out of 8 minimal winning coalitions. CREO, now integrated into BAN, increases its power compared to its initial condition in May 2021, given that BAN belongs to half of the minimal winning coalitions. Moreover, CREO also reduces its dependence on UNES, only present now in 1 out of 4 minimal winning coalitions to which BAN belongs.
\begin{table}[ht]
	\begin{center}
		\begin{tabular}{ l }
			\hline
			\textbf{Minimal winning coalitions} \\ \hline
			\{UNES, MUPP\}             \\
			\{UNES, BAN\}              \\
			\{UNES, ID, PSC\}          \\
			\{UNES, ID, IND\}          \\
			\{UNES, PSC, IND\}         \\
			\{MUPP, BAN, ID, PSC\}     \\
			\{MUPP, BAN, ID, IND\}     \\
			\{MUPP, BAN, PSC, IND\}    \\ \hline
		\end{tabular}
		
		\caption{Minimal winning coalitions for the weighted majority games associated with National Assembly of Ecuador in the period June-December 2021.}
		\label{Tb04}
	\end{center}
\end{table}

The Shapley-Shubik, Deegan-Packel, and Public Good power indices are not affected by the changes in the legislative bench structure given in Table \ref{Tb02}. This is because only the minimal winning coalitions, which did not change over time, are taken into account in the calculation of these	power indices. Table \ref{Tb05a} shows the distribution of power provided by these indices. Here, one can see that UNES benefits from these changes, as the power assigned by these indices is higher than in the initial situation in May 2021. However, the only power index that gives BAN (where CREO is integrated) more power than in the initial composition is the Shapley-Shubik power index, despite the fact that BAN now belongs to more minimal winning coalitions and decreases its dependence on UNES to form minimal winning coalitions. It is also important to note that the Shapley-Shubik power index gives equal power to the MUPP and BAN legislative benches. The ID, PSC and IND legislative benches are also given equal power. On the other hand, since the MUPP and BAN legislative benches participate in the same number of minimal winning coalitions and with the same cardinality, the Deegan-Packel power index also gives them equal power. The same holds for the legislative benches ID, PSC, and IND. Moreover, the Public Good power index gives the same power to all legislative benches except UNES since its calculation depends on the number of minimal winning coalitions and not on their cardinality.
\begin{table}
	\centering
	
	\begin{tabular}{r|cccccc}
		\hline
		              &             \multicolumn{6}{c}{Parties}             \\
		Power indices &  UNES  &  MUPP  &  BAN   &   ID   &  PSC   &  IND   \\ \hline
		        \ShSh & 0.4667 & 0.1667 & 0.1667 & 0.0667 & 0.0667 & 0.0667 \\
		          \DP & 0.2500 & 0.1562 & 0.1562 & 0.1458 & 0.1458 & 0.1458 \\
		          \PG & 0.2000 & 0.1600 & 0.1600 & 0.1600 & 0.1600 & 0.1600 \\ \hline
	\end{tabular}		
	
	\caption{Shapley-Shubik, Deegan-Packel, and Public Good power indices for the weighted majority games associated with National Assembly of Ecuador in the period June-December 2021.}
		\label{Tb05a}
\end{table}

On the other hand, as is shown in Table \ref{Tb05b}, the Colomer-Martínez and \HCM\ power indices are sensitive to changes in the structure of the legislative benches. Comparing these power allocation results with those obtained for the initial composition of the Assembly in May 2021, contrary to what happens with the power indices shown in Table \ref{Tb05a}, the power assigned to UNES almost does not change, while BAN doubles its power.

In conclusion, despite the movements that have occurred in the seven months of operation, UNES always gets the most power within the Assembly regardless of the power index used.

\begin{table}
	\centering
	
	\begin{tabular}{r r| c  c  c  c  c }
		\hline
		        \textbf{Parties} & \textbf{Power indices} & \textbf{Jun 21} & \textbf{Jul 21} & \textbf{12 Oct 21} & \textbf{26 Oct 21} & \textbf{Dec 21} \\ \hline
		\multirow{2}{1 cm}{UNES} & \CM                    & 0.4080          & 0.4016          & 0.4025             & 0.4036             & 0.4061          \\ 
		                         & \HCM                   & 0.4027          & 0.3950          & 0.3950             & 0.3950             & 0.3950          \\ \hline
		\multirow{2}{1 cm}{MUPP} & \CM                    & 0.1663          & 0.1602          & 0.1657             & 0.1652             & 0.1642          \\
		                         & \HCM                   & 0.1678          & 0.1613          & 0.1681             & 0.1681             & 0.1681          \\ \hline
		 \multirow{2}{1 cm}{BAN} & \CM                    & 0.1663          & 0.1663          & 0.1657             & 0.1712             & 0.1820          \\
		                         & \HCM                   & 0.1678          & 0.1681          & 0.1681             & 0.1748             & 0.1882          \\ \hline
		  \multirow{2}{1 cm}{ID} & \CM                    & 0.1047          & 0.1046          & 0.0928             & 0.0928             & 0.0930          \\
		                         & \HCM                   & 0.1074          & 0.1076          & 0.0941             & 0.0941             & 0.0941          \\ \hline
		 \multirow{2}{1 cm}{PSC} & \CM                    & 0.0929          & 0.0928          & 0.0928             & 0.0928             & 0.0930          \\
		                         & \HCM                   & 0.0940          & 0.0941          & 0.0941             & 0.0941             & 0.0941          \\ \hline
		 \multirow{2}{1 cm}{IND} & \CM                    & 0.0617          & 0.0744          & 0.0806             & 0.0744             & 0.0617          \\
		                         & \HCM                   & 0.0604          & 0.0739          & 0.0807             & 0.0739             & 0.0605          \\ \hline
	\end{tabular}
	
	\caption{Colomer-Martínez and \HCM\ power indices for the weighted majority games associated with National Assembly of Ecuador in the period June-December 2021.}
	\label{Tb05b}
\end{table}

\section{Conclusions and future work }\label{sec:conclusions}
In this paper, a definition of mergeability of weighted majority games is proposed. The first characterization of the Colomer-Martínez power index in the class of weighted majority games using mergeability and symmetry properties is also provided. Moreover, combining the ideas of Public Good and Colomer-Martínez power indices, a new power index for weighted majority games is also introduced, and one characterization is provided. Table \ref{cuadro06} shows a comparative summary of the characterizations of the power indices discussed in this paper. The last two columns are the characterizations proposed in this paper.
	\begin{table}[ht]
	\begin{center}
		\begin{tabular}{r|ccccc}
			\hline
			\multicolumn{1}{r|}{\textbf{Property}} & \textbf{\ShSh} & \textbf{\DP} & \textbf{\PG} & \textbf{\CM} & \textbf{\HCM} \\ \hline
			Efficiency &   \checkmark   &  \checkmark  &  \checkmark  &  \checkmark  &  \checkmark   \\
			Symmetry &   \checkmark   &  \checkmark  &  \checkmark  &              &               \\
			Null player &   \checkmark   &  \checkmark  &  \checkmark  &  \checkmark  &  \checkmark   \\
			Transfer &   \checkmark   &              &              &              &               \\
			\DP\ Mergeability &                &  \checkmark  &              &              &               \\
			\PG\ Mergeability &                &              &  \checkmark  &              &               \\
			Weighted symmetry &                &              &              &  \checkmark  &  \checkmark   \\
			\DP\ weighted mergeability &                &              &              &  \checkmark  &               \\
			\HCM\ weighted mergeability &                &              &              &              &  \checkmark   \\ \hline
		\end{tabular}
		
		\caption{Overview of the characterizations of the power indices.}\label{cuadro06}
	\end{center}
	
\end{table}

For future work, it would be interesting to propose generalizations of the Colomer-Martínez and \HCM\ power indices for games with a priori unions, consortia, or restricted communication. Moreover, it would also be interesting to study some methods to calculate these power indices.

\section*{Acknowledgment}
We would like to thank Balbina V. Casas-Méndez and two anonymous referees for their valuable comments. This work is part of the R+D+I project grants MTM2017-87197-C3-2-P, MTM2017-87197-C3-3-P, PID2021-124030NB-C32, and PID2021-124030NB-C33, that were funded by \newline MCIN/AEI/10.13039/501100011033/ and by ``ERDF A way of making Europe''/EU. This research was also funded by Grupos de Referencia Competitiva ED431C-2020/03 and ED431C-2021/24 from the Consellería de Cultura, Educación e Universidades, Xunta de Galicia.

\bibliography{biblio}
\bibliographystyle{spbasic2}
\end{document}